\newif\ifLIPICS
\LIPICSfalse

\newif\ifANON
\ANONtrue
\ANONfalse

\ifLIPICS
    \ifANON
         \documentclass[anonymous, english,cleveref,numberwithinsect]{lipics-v2021}
    \else
        \documentclass[a4paper,english,cleveref, autoref, thm-
        restate,numberwithinsect]{lipics-v2021}
    \fi

    \usepackage[T1]{fontenc}
    \usepackage{babel}
    \usepackage{float}
    \usepackage{amsmath}
    \usepackage{amsthm}
    \usepackage{amssymb}

    \usepackage{xcolor}
    \usepackage{graphicx} %
    
    \usepackage{tablefootnote}
    \usepackage{booktabs} %

    \DeclareMathOperator{\poly}{poly}
    \DeclareMathOperator{\polylog}{polylog}

    \newcommand{\M}{\mathcal{M}}

    \newcommand{\eps}{\varepsilon}

    \usepackage{mathtools}
    \usepackage{comment}
    \usepackage{etoolbox}   %
    \makeatletter
    \patchcmd{\thebibliography}{#1}{A99}{}{}   %
\else
    \documentclass[11pt,a4]{article}
    \usepackage{paper}
\fi

\title{Fast Nearest Neighbor Search for $\ell_p$ Metrics} 

\ifLIPICS
    \author{Robert Krauthgamer}{The Harry Weinrebe Professorial Chair of Computer Science, Weizmann Institute of Science, Rehovot, Israel}    {robert.krauthgamer@weizmann.ac.il}{https://orcid.org/0009-0003-8154-3735}
    {Work partially supported by the Israel Science Foundation grant \#1336/23.}

    \author{Nir Petruschka}{Weizmann Institute of Science, Rehovot, Israel}{nir.petruschka@weizmann.ac.il}{https://orcid.org/0009-0005-7051-6330}{}

\authorrunning{R. Krauthgamer and N. Petruschka} %

\acknowledgements{We thank Shay Sapir and Asaf Petruschka for helpful discussions.}

\relatedversion{}\relatedversiondetails
{\textcolor{red}{TODO}}{}



\Copyright{Robert Krauthgamer and Nir Petruschka}

\EventEditors{Hee-Kap Ahn, Michael Hoffmann, and Amir Nayyeri} 
\EventNoEds{3} \EventLongTitle{42nd International Symposium on Computational Geometry (SoCG 2026)}
\EventShortTitle{SoCG 2026} 
\EventAcronym{SoCG} \EventYear{2026} 
\EventDate{June 2--5, 2026} \EventLocation{New Brunswick, NJ, USA} \EventLogo{socg-logo.pdf} \SeriesVolume{367}
\ArticleNo{XX}

\else %

\ifANON
\author{Anonymous Authors}
\else
            \author{Robert Krauthgamer%
          \thanks{The Harry Weinrebe Professorial Chair of Computer Science.
          Work partially supported by the Israel Science Foundation grant \#1336/23.
            Email: \texttt{robert.krauthgamer@weizmann.ac.il}
          } 
          \qquad 
        Nir Petruschka\thanks{
            Email: \texttt{nir.petruschka@weizmann.ac.il}
        } 
        \\  Weizmann Institute of Science }
\fi %

\fi %

\begin{document}

\maketitle

\begin{abstract}
The Nearest Neighbor Search (NNS) problem asks to design a data structure that preprocesses an $n$-point dataset $X$ lying in a metric space $\M$,
so that given a query point $q \in \M$, one can quickly return a point of $X$
minimizing the distance to $q$.
The efficiency of such a data structure is evaluated primarily by the amount of space it uses and the time required to answer a query.
We focus on the fast query-time regime,
which is crucial for modern large-scale applications,
where datasets are massive and queries must be processed online,
and is often modeled by query time $\poly(d \log n)$.
Our main result is such a randomized data structure for NNS in $\ell_p$ spaces, $p>2$, 
that achieves $p^{O(1) + \log\log p}$ approximation with fast query time and $\poly(dn)$ space.
Our data structure improves, or is incomparable to, %
the state-of-the-art for the fast query-time regime 
from [Bartal and Gottlieb, TCS 2019] and [Krauthgamer, Petruschka and Sapir, FOCS 2025]. 
\end{abstract}

\section{Introduction}

The Nearest Neighbor Search (NNS) problem asks to design a data structure (also called a scheme) that preprocesses an $n$-point dataset $X$ lying in a metric space $\M$,
so that given a query point $q \in \M$, one can quickly return a point of $X$
minimizing the distance to $q$ (or approximately minimizing it in the approximate version).
The efficiency of such a data structure is evaluated primarily by the amount of space it uses and the time required to answer a query. The preprocessing time is a secondary measure and is usually comparable to the space usage.
Because of its central role in areas such as machine learning,
data analysis, and information retrieval,
NNS has been the subject of extensive research, both practical and theoretical
(see, e.g., the surveys~\cite{AndoniIndyk17, AIR19}).

It is well known that approximate NNS can be reduced to solving
$\polylog(n)$ instances of the approximate \emph{near} neighbor problem~\cite{HIM12}.  
For this reason, we restrict attention to the latter.
\begin{definition} %
The Approximate Near Neighbor problem for a metric space \((\M, d_\M)\)
and parameters $c \geq 1$, $r>0$, abbreviated \emph{\((c,r)\)-ANN},
is the following. 
Design a data structure that preprocesses an \(n\)-point subset \(X \subseteq \M\), 
so that given a query \(q \in \M\) with $d_\M(q,X) \leq r$,%
\footnote{If $d_\M(q,X) > r$, it may report anything,
  where as usual, $d_\M(q,X)\coloneqq \min_{x^* \in X} d_\M(x^*, q)$. 
}
it reports \(x \in X\) such that
\[
  d_\M(q, x) \leq cr.
\]
In a randomized data structure,
the reported $x\in X$ satisfies this with probability at least $2/3$. 
\end{definition}

We focus on the fast query-time regime,
which is crucial for modern large-scale applications where datasets are massive and queries must be processed online,
and is often modeled by query time $\poly(d \log n)$. 
In $\ell_p$ spaces, ANN in this regime is well understood for $1 \leq p \leq 2$~\cite{IM98, KOR00, HIM12} and for $p=\infty$~\cite{Indyk01_ell_infty, ACP08, KP12a}. 
For $2<p<\infty$, the situation is less clear: 
there exists a handful of data structures, each suitable for a different range of $p$, 
as detailed in \cref{tbl:ANN-Comparision}. 
We present a new scheme for ANN in $\ell_p$, $p>2$, with fast query time, 
that offers an improved tradeoff between approximation and space, as follows.

\begin{theorem}
\label{thm:ANN-for-l_p}
Let $p>2$, $d\ge 1$.
Then for every 
$r>0$,
there is a randomized data structure for $(c,r)$-ANN in $\ell_p^d$, 
where $c=p^{O(1) 
+ \log\log p}$, 
that has query time $\poly(d\log n)$ 
and preprocessing (both space and time) 
$\poly(dn)$.%
\footnote{The time and space bounds are independent of $p$.}
\end{theorem}

Studying ANN in $\ell_p$, $p>2$, is important both practically and theoretically.
Real-world data and applications may motivate norms 
that emphasize outliers, e.g., for anomaly detection, 
or alter the presence of ``hubs'', e.g., to affect classification;
such data structures were indeed used for time-series classification~\cite{YYW11},
see \cite{BG19} for additional references. 
From a theoretical perspective, the geometry of $\ell_p$ spaces 
undermines existing algorithmic techniques and requires developing new ones. 
A key challenge is to bridge between $p=2$ and $p=\infty$.
In the fast query-time regime, this means interpolating between 
the classical $(O(1), r)$-ANN in $\ell_2$~\cite{HIM12} 
and the $(O(\log\log{d}), r)$-ANN in $\ell_\infty$~\cite{Indyk01_ell_infty}, 
which both use only $\poly(dn)$ space. 
It is natural to conjecture that $2<p<\infty$ exhibits an interpolation between these two guarantees, 
and since $\ell_{\infty}^d$ is $O(1)$-equivalent to $\ell_{\log d}^d$ by Hölder's inequality, 
this interpolated data structure is conjectured 
to achieve $O(\log p)$ approximation using $\poly(nd)$ space. 
The first step towards this conjecture, in \cite{AIK09, Andoni09_thesis}, 
devised a reduction from $\ell_p$ to $\ell_\infty$, 
and obtained a data structure with $O(\log^{1/p} n \log\log d)$-approximation, 
which is mainly suited for large values of $p$. 
A nontrivial reduction, devised in \cite{BG19},
reduced $\ell_p$ to $\ell_2$ and obtains $2^{O(p)}$-approximation, 
which is a major improvement for small values of $p$, 
although it is doubly-exponentially worse than the conjecture.
A more sophisticated reduction, that was devised recently in \cite{KPS25}, 
achieves approximation $\poly(p)$, which is an exponential improvement.
However, it goes through multiple intermediate $\ell_t$ spaces ($2 \leq t < p$) via a recursive argument
that increases the space to $n^{O(\log p)}$, much higher than conjectured. 
Our \cref{thm:ANN-for-l_p} essentially completes the improvement of \cite{KPS25}, 
by decreasing the space complexity back to the conjectured $\poly(dn)$, 
albeit slightly increasing the approximation to $p^{\log\log p}$.
In particular, it resolves a question posed in \cite{KPS25}, 
of whether the recursion can avoid this higher space complexity.

\begin{table}[t]
\begin{center}
\begin{tabular}{lll}
\midrule
Approximation  & Space & Reference  \\
\midrule
\midrule
$O(\log^{1/p}{n}\log\log d)$  & $\poly(nd)$ & \cite{AIK09, Andoni09_thesis}  \\
$2^{O((\log{d})^{2/3}(\log\log{d})^{1/3})}$  & $\poly(nd)$ & \cite{ANNRW18b}  \\
$2^{O(p)}$     & $\poly(nd)$ & \cite{BG19} \\
$p^{O(1)}$     & $n^{O(\log p)}$ & \cite{KPS25}  \\
$c$     &  $n^{ O(p/c) \cdot \log\log n}$ & \cite{BBMWWZ24}  \\
$p^{O(1) + \log\log p}$    & $\poly(nd)$     & \cref{thm:ANN-for-l_p} \\
\midrule
\end{tabular}
\end{center}
\caption{Known data structures for ANN in $\ell_p$, $p>2$, in the fast query time regime.}
\label{tbl:ANN-Comparision}
\end{table}

The proof of \cref{thm:ANN-for-l_p} is based on 
a simple yet powerful enhancement of the known NNS schemes from \cite{BG19, KPS25}, 
which utilizes classical results from \cite{AP90, ABCP98}  about sparse covers. 
We provide an overview of the algorithms of \cite{BG19, KPS25} in \cref{sec:Overview}, 
along with an intuitive explanation of our approach at the beginning of \cref{sec:ANN}.

\subsection{Related Work}
Many of the existing results on approximate nearest neighbor search in $\ell_p$ spaces
focus on the case $1 \leq p \leq 2$~\cite{KOR00,IM98,HIM12, DIIM04, AI06, Andoni09_thesis, Ngu13}.
In this setting, $O(1)$-approximation can be achieved with $\poly(dn)$ preprocessing (space and time),
and query time polynomial in $d \log n$~\cite{KOR00, HIM12}.

In recent years, significant progress has been made for the case $p \in (2,\infty)$,
and current results can be broadly divided into three regimes.  
The first regime consists of data structures that achieve
moderate approximation and query time using near-linear space~\cite{Andoni09_thesis, ANNRW18a, ANNRW18b, ANRW21, KNT21, AN25}. 
The second regime has small approximation factor, say $O(1)$ or even $1+\eps$,
in which case both the query time and the preprocessing requirements (space and time)
are typically very large~\cite{BG19, BBMWWZ24}.
The third regime, which is the focus of our work,
has fast query time, namely, $\poly(d \log n)$,
and existing results either
achieve $2^{O(p)}$-approximation with $\poly(dn)$ space~\cite{BG19},
or better approximation $o(2^p)$ at the cost of a much bigger $n^{\omega(1)}$ space~\cite{BBMWWZ24, KPS25}.
Our result is the first to obtain both $o(2^p)$ approximation and $\poly(dn)$ space. 
\cref{tbl:ANN-Comparision} provides a comparison of all known data structures in this regime.

We point out that the recursive embedding technique from \cite{KPS25}, which is used here, 
has been applied successfully also to other problems involving $\ell_p$ spaces,
such as the construction of Lipschitz decompositions~\cite{KP25, KPS25, NR25},
geometric spanners~\cite{KP25, KPS25},
and low-distortion embeddings~\cite{KPS25, NR25}.

\section{Preliminaries}
Given a metric space $\M=(X,d_\M)$, we denote by $B_{d_\M}(x,r) \coloneqq \{y \in X : d_\M(x,y)\leq r\}$ 
the ball of radius $r>0$ centered at a point $x \in \M$.

For every \( p, q \in [1, \infty) \),
the Mazur map  $M_{p,q}: \ell^d_p \to \ell^d_q$ is computed by taking,
in each coordinate, the absolute value raised to power \( p/q \),
but keeping the original sign.
Our algorithm crucially relies on the following property of this map.

\begin{theorem} [\cite{benyamini1998geometric,BG19}]\label{thm:mazur_map}
Let $1 \leq q < p < \infty$ and $C_0>0$,
and let $M$ be the Mazur map $M_{p,q}$ scaled down by factor $\frac{p}{q}{C_0}^{p/q-1}$.
Then for all $x,y\in\ell_{p}^d$ such that $x,y \in B_{\ell_p}(0,C_0)$,
\[
  \tfrac{q}{p} (2 C_0 )^{1 - p/q} ||x - y||_{p}^{p/q}
  \leq ||M(x) - M(y)||_{q}
  \leq ||x - y||_{p} .
\]
\end{theorem}

\section{Overview of Algorithms from \cite{BG19, KPS25}} \label{sec:Overview}

In this section, we review the algorithms of \cite{BG19} and \cite{KPS25} for ANN in $\ell_p$, $p>2$.
For the rest of the section, fix a dataset $X \subset \ell_p^d$ with $|X|=n$ for some $p>2$.

\subsection{$(2^{O(p)}, r)$-ANN with $\poly(dn)$ Space \cite{BG19}}
\label{sec:Overview_BG19}

In the preprocessing stage, consider a set of $k=\frac{p}{2}\cdot O(\log\log d)$ possible approximation factors $\mathcal{\hat C}=\{\hat c_i\}_{i=0}^{k}$, where $\poly(d)=\hat c_0 \geq \hat c_1 \geq \dots \geq \hat c_k = 2^{O(p)}$.
First, compute for $X$ an initial NNS data structure $A_{\text{init}}$ using \cite{Chan98},
which provides approximation $\hat c_0=\poly(d)$
using query time $\poly(d\log{n})$ and space $\poly(d)\Tilde{O}(n)$.%
\footnote{Throughout, the notation $\tilde{O}(f)$ hides factors that are logarithmic in $f$.}
Then, for every data point $x\in X$ and every approximation $\hat c \in \mathcal{\hat C}$,
compute a (scaled) Mazur map $M_{x,\hat c}: \ell_p^d\to \ell_2^d$ for the points in the set $B_{\ell_p}(0, \hat cr) \cap (X-x)$.
Finally, compute for the image points in $\ell_2^d$
the $(2, r)$-ANN data structure $A_{x,\hat c}$ from \cite{HIM12},
which uses $\poly(d\log{n})$ query time and $\poly(dn)$ space.
We have this data structure for each point $x \in X$ and each approximation factor $\hat c \in \mathcal{\hat C}$,
and clearly $|B_{\ell_p}(x, \hat cr)| \leq n$,
hence the total space requirement is $O(p \cdot \log\log d) n \cdot \poly(dn) = \poly(dn)$.

At query time, given a query point $q \in \ell_p^d$,
find a $\hat c_0$-approximate solution $x_0$ using $A_{\text{init}}$.
The crucial observation is that since the Mazur map ensures a distortion
that depends on the diameter of the point set (\cref{thm:mazur_map}),
the answer from $A_{x_0. \hat c_0}$ is a $\hat c_1$-approximate solution $x_1$.
Applying this procedure iteratively, the approximation factor decreases even faster than geometrically,
roughly as $\hat c_i = \hat c_{i-1}^{1-2/p}$. 
Hence, after $k=O(p \log\log{c_0})$ iterations 
we obtain an approximate solution $x_k$ with $\hat c_k = 2^{O(p)}$, 
where the approximation factor does not improve further.  

\subsection{$(\poly(p), r)$-ANN with $\poly(d)n^{O(\log p)}$ Space \cite{KPS25}}

In \cite{KPS25}, the image space of the Mazur map is changed from $\ell_2$ to $\ell_t$ for general $1 \leq t < p$.
Generalizing the results from \cite{BG19},
a $(c_t, r)$-ANN data structure in $\ell_t^d$ with query time $Q(n)$ and space $S(n)$
is used to construct a $(2^{O(p/t)} \cdot c_t, r)$-ANN data structure for $\ell_p^d$
with query time $O(d) + \frac{p}{t}O(\log\log d)Q(n)$ and space $\frac{p}{t}O(\log\log d) n \cdot S(n)$.
Using the above result with $t=p/2$, and applying it recursively to decrease $p$ to $2$ 
(which is actually a double recursion, because we also iterate over the $\hat c_i$'s), 
yields a $(\poly(p), r)$-ANN data structure with query time $\poly(d\log{n})$.
The caveat is that every application of the recursive step 
multiplies the space of the data structure by factor $n$, 
which yields a data structure with space $\poly(d)n^{O(\log p)}$.

\section{$(2^{\tilde{O}(\log p)}, r)$-ANN with $\poly(dn)$ Space} \label{sec:ANN}

In this section, we give the proof of \cref{thm:ANN-for-l_p}. 
We first explain the intuition, 
and for simplicity we restrict this discussion to reducing the space requirement of \cite{BG19}; 
reducing the space requirements of \cite{KPS25} is similar in spirit, although more technical.

Revisiting \cref{sec:Overview_BG19}, the final solution $x_k$ is obtained 
by finding iteratively a sequence of intermediate solutions $x_0, x_1, \ldots, x_{k-1}$. 
Each iteration $i<k$ makes progress by finding a point $x_i$ 
and restricting the search region to $B_{\ell_p}(x_i, \hat c_{i}r)$, which has bounded diameter, 
and thus applying a Mazur map on this region has distortion guarantees. 
It follows that querying the data structure $A_{x_i,\hat c_i}$ (computed over $B_{\ell_p}(x_i,\hat c_i r) \cap X$) 
finds a point $x_{i+1}$ and we can restrict the search region even further, 
to diameter $\hat c_{i+1}r$. 

The preprocessing phase prepares for the possibility that 
each point $x \in X$ will serve (at query time) the $\hat c_i$-approximate solution,
i.e., the search region will be restricted to $B_{\ell_p}(x, \hat c_i r)$. 
To make progress and restrict the search region even further, 
a data structure $A_{x,\hat c_i}$ is constructed for (the points in) this region.
Our key idea in \cref{thm:ANN-for-l_p} is that, 
rather than preparing a \emph{separate} data structure for each search region, 
the algorithm constructs one \emph{global} collection of data structures 
that together cover all the possible search regions.
For every $\hat c_i$, the algorithm constructs a set of ANN data structures 
computed on a collection of subsets $\mathcal{S} \subseteq 2^{X}$, such that 
for every point $x \in X$ there is some $S \in \mathcal{S}$ 
that contains the search region $B_{\ell_p}(x, \hat c_i r) \cap X$.
In addition, every $S\in\mathcal{S}$ has diameter at most $\beta \hat c_i r$ for some $\beta>1$. 
We also want the total number of points in $\mathcal{S}$ (counting repetitions) to be small. 
The preprocessing algorithm simply stores for every $x \in X$ a reference to 
a set $S_x \in \mathcal{S}$ with $B_{\ell_p}(x, \hat c_i r) \cap X \subseteq S_x$, 
and at query time, if $x$ serves as a $\hat c_i$-approximate solution, 
the algorithm queries the ANN data structure constructed for $S_x$.
Since $S_x$ has a diameter at most $\beta \hat c_i r$, 
this will still cause the search region's diameter to shrink in the next iteration (although by a slightly smaller factor). 
Since the total number of points in $\mathcal{S}$ is small, 
the total memory used by all the ANN data structures will be small too. 

It remains to show that the preprocessing phase can indeed find efficiently 
a collection of subsets of $X$ with the above properties. 
Fortunately, this was shown to be possible in \cite{AP90, ABCP98},
and has become a fundamental algorithmic tool with numerous applications 
in distributed computing, network design, routing, graph algorithms, and metric embeddings.

\begin{definition}[Sparse Neighborhood Cover~\cite{AP90, ABCP98}] 
A \emph{$(\beta,r)$-sparse cover} of a metric space $\M = (X,d_\M)$ 
is a collection of subsets (called clusters) $\mathcal{S} \subseteq 2^{X}$, 
each of diameter at most $\beta r$, 
such that for every $x \in X$ there exists $S \in \mathcal{S}$ with $B_{d_\M}(x,r) \subseteq S$.
The total number of points $\sum_{S\in\mathcal{S}}|S|$ is called the \emph{sparsity} of $\mathcal{S}$. 
\end{definition}

\begin{theorem}[\cite{ABCP98}]\label{thm:Sparse-Cover} 
There is an algorithm that, given an $n$-point metric space $\M$ and parameters $\beta>1$ and $r > 0$, 
outputs a $(\beta, r)$-sparse cover of $\M$ of sparsity $O(n^{1+1/\beta})$, 
and runs in $O(\beta n^{2+2/\beta})$ time. 
\end{theorem}

We are now ready to prove \cref{thm:ANN-for-l_p}, 
largely following the proof structure of \cite[Theorem 1.8]{KPS25}. 

\begin{proof}[Proof of \cref{thm:ANN-for-l_p}]
Let $X \subset \ell_{p}^d$ be an $n$-point dataset for some $p\in(2,\infty)$.  
For clarity of exposition, we assume that $p$ is a power of $2$,
which can be easily resolved, see~\cite[Theorem 1.2]{KPS25}. 
Also, by an application of Hölder's inequality, we may assume that $p \leq \log d$.

We construct the ANN scheme using a doubly-recursive procedure.
The first recursion assumes access to an ANN scheme for $\ell_p^d$
that achieves approximation factor $c_{\text{base}}$, 
and provides a new ANN scheme for $\ell_p^d$
that achieves improved (smaller) approximation $c_{\text{new}}$.
This step crucially relies on access to yet another ANN data structure, 
for $\ell_t^d$, for $t = p/2$, that is actually constructed by the same method.
This leads to a second recursion, of constructing ANN data structures 
for intermediate spaces $\ell_p^d, \ell_{p/2}^d, \ldots, \ell_{2}^d$,
where the space $\ell_2^d$ is known to have ANN data structures with $O(1)$ approximation.

We next describe the first recursion, 
i.e., how to construct an improved $(c_{\text{new}}, r)$-ANN scheme for $\ell_p^d$ 
given a $(c_\text{base}, r)$-ANN scheme for $\ell_p^d$ 
and a $(c_t, r)$-ANN scheme for $\ell_t^d$, where $t < p$.
In the preprocessing phase, use \cref{thm:Sparse-Cover} to construct for $X$ 
a $(\beta, 2 c_\text{base} r)$-cover $\mathcal{S}$ 
with sparsity $\tilde{O}(n^{1+1/\beta})$ for $\beta = \log p$. 
During the construction of $\mathcal{S}$, 
store for every $x \in X$ a reference to a set $S_x\in\mathcal{S}$ that ``covers'' it, 
i.e., $B_{\ell_p}(x,2 c_\text{base} r) \cap X \subseteq S_x$,
which is guaranteed to exist in a sparse cover. 
In addition, for every $S \in \mathcal{S}$ designate (arbitrarily) a center point $y \in S$,
apply a Mazur map $M^{y}: \ell_{p}^{d} \to \ell_{t}^{d}$ 
scaled down by factor $\frac{p}{t} \cdot (2\beta c_\text{base} r)^{p/t-1}$
on $B_{\ell_p}(0,2 \beta c_\text{base} r) \cap (X-y)$,
and construct for these image points a $(c_t, r)$-ANN scheme $A_S$. 
Finally, construct a $(c_\text{base},r)$-ANN scheme $A_\text{base}$ for $X$, 
and amplify the success probabilities of both data structures to $5/6$ 
by the standard method of independent repetitions.
Given a query $q\in \ell_{p}^d$ that is guaranteed to have $x^* \in X$ with $\|x^* - q\|_p \leq r$, 
query $A_{\text{base}}$ for the point $q$ and obtain an answer $x_{\text{base}} \in X$. 
Then find its cluster $S_{x_{\text{base}}}$ and this cluster's designated center $y$, 
query $A_{S_{x_{\text{base}}}}$ for the point $M^{y}(q-y)\in \ell_t^d$, 
and use its answer $M^{y}(z_{\text{out}}-y) \in M^y(X-y)$ 
to output the corresponding $z_{\text{out}} \in X$. 

The next claim is analogous to \cite[Claim 4.1]{KPS25},
and the main difference is using the sparse cover. 

\begin{claim}\label{claim:ANN_Mazur_Map}
With probability at least $2/3$, 
we have $\|z_{\text{out}}-q\|_p \leq c_{\text{new}} r$, 
where $c_{\text{new}}=(\tfrac{p}{t})^{t/p} \ c_t^{t/p} \ (4\beta c_\text{base})^{1-t/p}$.
\end{claim}
\begin{proof}
With probability at least $\frac{5}{6}$, 
the data structure $A_{base}$ outputs a point $x_{\text{base}}$ with $\|x_{\text{base}}-q\|_p \leq c_\text{base} r$. 
Let $S_{x_{\text{base}}}$ be the set in the cover referenced by $x_{\text{base}}$, 
and let $y \in S_{x_{\text{base}}}$ be its designated center point.
Since 
\[
  \|x^*-x_{\text{base}}\|_p\leq \|x^*-q\|_p+\|q-x_{\text{base}}\|_p\leq 2c_\text{base} r,
\]
we get that $x^* \in B_{\ell_p}(x_{\text{base}}, 2c_\text{base} r) \cap X \subseteq S_{x_{\text{base}}}$. 
Observe that by \cref{thm:mazur_map}, $\|M^y(x^*)-M^y(q)\|_t\leq r$,
and thus with probability at least $\frac{5}{6}$, 
querying $A_{S_{x_{\text{base}}}}$ finds a point $M^{y}(z_\text{out}) \in M^{y}(S_{x_\text{base}}) \subseteq M^{y}(X)$ with $\|M^y(z_\text{out})-M^y(q)\|_t\leq c_t r$. 
Applying a union bound, we see that with probability at least $2/3$, both events hold. 
In this case, we have by \Cref{thm:mazur_map}, that
\[
    \frac{t}{p} \cdot {(4\beta c_\text{base} r)^{1-p/t}} \cdot \|z_{\text{out}}-q\|_{p}^{p/t} 
    \leq \|M^{y}(z_{\text{out}})-M^{y}(q)\|_t 
    \leq c_t r,
\] 
and by rearranging, we obtain 
$\|z_{\text{out}}-q\|_p \leq (\tfrac{p}{t})^{t/p} \ c_t^{t/p} \ (4\beta c_\text{base})^{1-t/p} r = c_{\text{new}} r$.
\end{proof}
    
Denote by $\hat{c}_i$ the approximation of the ANN scheme obtained by $i$ applications of \cref{claim:ANN_Mazur_Map},
where the initial ANN scheme is the one from \cite{Chan98}, 
with approximation $\hat{c}_0=\poly(d)$. 
We also denote by $c_t$ the approximation of an ANN scheme for $\ell_t$,
that is constructed by the same method (i.e., recursively),
except that for $\ell_2^d$ we use a $(2,r)$-ANN scheme from \cite{HIM12} 
with $\poly(d\log n)$ query time and $\poly(dn)$ space and preprocessing time. 
Using \cref{claim:ANN_Mazur_Map} with $t=p/2$,
and furthermore applying this recursively $k=\lceil \log(\log \hat{c}_0) \rceil=O(\log\log d)$ times, 
we obtain
\[
    \hat{c}_k 
    \leq \sqrt{8\beta c_{p/2} \hat{c}_{k-1}} 
    \leq \sqrt{8\beta c_{p/2} \sqrt{8\beta c_{p/2}\hat{c}_{k-2}}} 
    \leq \ldots \leq 8\beta c_{p/2}\cdot \hat{c}_0^{1/2^{k}} 
    \leq 16\beta c_{p/2}, 
\]
i.e., this scheme has approximation $c_p = \hat{c}_k \leq 16 \beta c_{p/2}$. 
We can amplify the success probability of this scheme to $1 - \frac{1}{3\log p}$ 
by the standard method of $O(\log\log p) = O(\log\log\log d)$ independent repetitions. 
Now by recursion over $p$ for $\log p$ levels, we get that 
\[
c_p \leq \left( 16 \beta \right)^{\log p} 
     = \left(16 \log p \right)^{\log p} 
     = p^{4 + \log\log p},
\]
and the overall success probability is at least $\tfrac{2}{3}$ by a union bound.

We are left to analyze the query time of the algorithm, and its space and preprocessing time.
Each level of the second recursion makes a total of $k \cdot O(\log\log p) = \tilde{O}(\log\log{d})$ calls 
to an ANN scheme for $\ell_t$, for different intermediate values of $t$. 
Since the $(2,r)$-ANN for $\ell_2$ from \cite{HIM12} has query time $\poly(d \log n)$, 
and recalling that $p\leq \log d$, 
the overall query time is 
$\tilde{O}(\log\log{d})^{\log{p}} \cdot \poly(d\log n)=\poly(d\log n)$. 

To analyze the space and preprocessing time, we prove the following claim.

\begin{claim}\label{claim:Space-Analysis}
  There exists an absolute constant $D>2$ such that  
  when the data structure for $\ell_t^d$, $t=2^i$, is computed on $m$ points, 
  it uses total space and preprocessing time $\Tilde{O}(\log\log d)^{i} \poly(d)O(m)^{D + i/\log p}$.
\end{claim}
\begin{proof}
We only analyze the space usage of the data structure; 
the analysis of the preprocessing time follows similarly, 
as it takes $O(d)$ time to compute a Mazur map 
and $O(\log p \cdot m^{2(1+1/\log p)}) \leq O(\log\log d \cdot m^{2(1+1/\log p)})$ time to compute a sparse cover.

The proof proceeds by induction on $i \geq 0$.
For $i=0$, the claim follows because the $(2,r)$-ANN from \cite{HIM12}, 
when computed on $m$ points, 
uses at most $\poly(d)m^{D}$ space for some absolute constant $D>1$. 
Now, assume the claim holds for $i-1 \geq 0$. 
The ANN scheme at level $i$ of the recursion consists of two types of ANN schemes. 
The first type is an ANN scheme from \cite{Chan98} computed on all $m$ points, 
which uses $\poly(d)\tilde{O}(m) \leq \poly(d)O(m)^{D(1+1/\log p)}$ space. 
The second type are multiple ANN schemes at level $i-1$ 
that are computed on different subsets of the $m$ points.  
For every $j=0,1,\ldots,k=\tilde{O}(\log\log d)$, 
let $\mathcal{S}^j$ be the 
sparse cover of sparsity $O(m^{1+1/\log p})$ 
computed for the points at the $j$-th level of the first recursion. 
For every level $j$ of the first recursion and cluster $S\in \mathcal{S}^j$, 
the algorithm computes a data structure of level $i-1$ on $S$. 
By the induction hypothesis, the space of this data structure is bounded by $\Tilde{O}(\log\log d)^{i}\poly(d) O(m)^{D+i/\log p}$. Observe that the function $f:x \mapsto \Tilde{O}(\log\log d)^{i}\poly(d)x^{D + i/\log p}$
satisfies that %
$f(a)+f(b) \leq f(a+b)$ for all $a,b \geq 1$,
and, since every data structure is constructed over some $S \in \mathcal{S}^j$ containing at most $m$ points, we can bound the contribution of every point $p \in S$ to the total space usage of the data structure by $\frac{f(m)}{m}$. Since for every $j$ the sparsity of $\mathcal{S}^j$ is at most $O(m^{1+1/\log p})$, the total space usage of all ANN schemes computed over all clusters of $\mathcal{S}^j$ is at most
\[
O(m^{1+1/\log p})\frac{f(m)}{m}\leq \Tilde{O}(\log\log d)^{i} \poly(d)O(m)^{D+(i+1)/\log p}.
\]
Hence, the total space usage is 
\[
\begin{aligned}
&\Tilde{O}(\log\log d)\cdot\left(\Tilde{O}(\log\log d)^{i}\cdot \poly(d)\cdot O(m)^{D+(i+1)/\log p}\right) + \poly(d)\cdot \Tilde{O}(m) \\
& \leq \Tilde{O}(\log\log d)^{i+1}\cdot \poly(d)\cdot O(m)^{D+(i+1)/\log p},
\end{aligned}
\]
completing the proof.
\end{proof}
Finally, we use \cref{claim:Space-Analysis} for $m=n$ and $i=\log p$, 
and obtain that the space usage of the entire ANN data structure is bounded by 
\[
\Tilde{O}(\log\log d)^{\log p}\poly(d)\cdot n^{D+\log p/\log p} \leq \poly(d) \cdot n^{D+1} \leq \poly(dn),
\]
which completes the proof of \cref{thm:ANN-for-l_p}. 
\end{proof}

\begin{remark}
    Modifying the parameter $\beta$ of the sparse cover in the proof of \cref{thm:ANN-for-l_p} from $\log p$ to $\tfrac{\log p}{\delta}$ for $0<\delta<1$ yields a data structure with a slightly larger approximation factor $p^{O(1)+\log(1/\delta)+\log\log p}$, but space requirement that matches that of \cite{HIM12} up to subpolynomial factors in $d$ and an additional $n^{\delta}$ term.  
\end{remark}

\begin{remark}\label{rem:ANN_General_Norms}
    The same technique used in the proof of \cref{thm:ANN-for-l_p}, namely applying \cref{thm:Sparse-Cover} to construct covers in the preprocessing phase, can also be used to improve the space requirement of ANN for general normed spaces from \cite[Theorem~3]{ANNRW18b}. More specifically, for every $0 < \delta < 1$, one can shave an $\Omega(n^{1-\delta})$ factor from the space of their data structure, at the cost of an additional $O(\delta^{-1})$ factor in the approximation.
\end{remark}

\ifLIPICS
\else
  \paragraph*{Acknowledgments.}
  We thank Shay Sapir and Asaf Petruschka for helpful discussions.
\fi

{\small
  \bibliographystyle{alphaurl}
  \bibliography{references}
} %

\end{document}